\documentclass[letterpaper,conference]{IEEEtran}

\usepackage{amsmath}
\usepackage{amssymb}
\usepackage{amsthm, graphicx}
\usepackage{algorithm,algorithmic}

\newtheorem{prop}{Proposition}
\newtheorem{lem}[prop]{Lemma}
\newtheorem{thm}[prop]{Theorem}
\newtheorem{cor}[prop]{Corollary}

\newtheorem{defi}[prop]{Definition}

\newtheorem{ex}[prop]{Example}
\newtheorem{rem}[prop]{Remark}

\newcommand{\Z}{\mathbb{Z}}
\newcommand{\N}{\mathbb{N}}

\IEEEoverridecommandlockouts
\renewcommand{\baselinestretch}{0.97}

\title{Cross-Error Correcting Integer Codes over $\Z_{2^m}$}
\author{\IEEEauthorblockN{Anna-Lena Trautmann\thanks{ALT is also with the Department of Electrical and Electronic Engineering, University of Melbourne. She was supported by Swiss National Science Foundation  Fellowship no. 147304.}  and Emanuele Viterbo} \IEEEauthorblockA{Department of Electrical and Computer Systems Engineering, Monash University, Australia.}}

\begin{document}

\maketitle

\begin{abstract}
In this work we investigate codes in $ \Z_{2^m}^n$ that can correct errors that occur in just one coordinate of the codeword, with a magnitude of up to a given parameter $t $. 
We will show upper bounds on these \emph{cross codes}, derive constructions for linear codes and respective decoding algorithm. The constructions (and decoding algorithms) are given for length $n=2$ and $n=3$, but for general $m$ and $t $.
\end{abstract}

\section{Introduction}

To define codes over a set of integers is a well-known concept useful e.g.\ in coded modulation and magnetic recording. A linear integer code $C\subseteq \Z_q^n$ can be defined via a parity check matrix $H\in \Z_q^{N\times n}$ as (see e.g.\ \cite{vi98})
\[C= \{ v \in \Z_q^n \mid v H^T = 0\} .\]
Depending on the application different error models may apply and therefore different metrics can be used for constructing integer codes. In this work we want to investigate \emph{cross errors of magnitude $t $}, i.e.\ error vectors of the type $\alpha e_i$ where $e_i$ it the $i$-th unit vector and $\alpha \in \{-t,-t+1, \dots, t-1,t\}$. This type of error is a special case of the error type in \cite{vi98} and are a generalization of the definition of cross errors in \cite{mo07p}. Moreover, cross error correcting integer codes can be used for single peak-shift correction \cite{le93,ta97p}. The code constructions known for these types of errors are mainly over $\Z_q$ for odd $q$, whereas many applications (such as QAM) suggest that codes over $\Z_{2^m}$ would be of interest. This is why we investigate cross-error correcting integer codes (also called \emph{cross codes}) over $\Z_{2^m}$ in this work. Note that this is one of the open problems stated in \cite{vi98}.

For simplicity we define the absolute value of $x\in \Z_{2^m}$ as $|x|:=\min\{x, 2^m-x\}$. For $v,w \in \Z_{2^m}^n$ the Lee distance $d_L$ is defined as $d_L(v,w)= \sum_{i=1}^n |v_i -w_i |$. The Lee weight is defined analogously. 
One can easily see that codes that can correct errors of Lee weight at most $t$ are also cross codes, able to correct cross errors of magnitude up to $t$:
\begin{thm}
Every $t$-error correcting Lee code in $\Z_{2^m}^n$ is also a cross-error correcting code with magnitude $t$ in $\Z_{2^m}^n$.
\end{thm}
Codes for the Lee metric are well-known and have extensively been studied, e.g.\ in \cite{et11a,et10p,go70a} and references therein.  Even though, again not much is known for codes over $\Z_{2^m}^n$. In \cite{na79} a construction for $t$-Lee-error correcting codes over $\Z_{2^m}$ is given for $t=1,2$ but this construction is restricted to only certain sets of parameters.

We denote Lee-metric codes by $C^L$ and cross-error correcting codes by $C^+$.
If they are linear we denote them by $C^L_{lin}$ and $C^+_{lin}$, respectively.
For non-linear integer codes one can easily find examples where one can get cross codes with a larger cardinality than possible for Lee-metric codes.
\begin{ex}
The largest possible $2$-error correcting Lee code in $\Z_8^2$ has cardinality $4$, e.g.\ $C^L = \{(0,0),(1,4),(4,2),(5,6)\}$. But the code $C^+=\{(1,0),(4,1),(6,6),(0,3),(3,4)\}$ is a cross code with error magnitude $2$ with $5$ elements.
\end{ex}
This further motivates the interest in studying not only Lee metric codes but specifically cross codes over $\Z_{2^m}^n$.

The paper is structured as follows. First we will derive some bounds and compare them to the bounds for Lee-metric codes of the same parameters. Then we will derive code constructions and present decoding algorithms for these codes.


\section{Metric and sphere packing for cross errors}

Usually in coding theory one defines a metric according to the error model one has. For the cross error model this is not straight-forward but we can define the following \emph{cross distance} on $\Z_q^n$.
\begin{defi}
For any $v,w \in \Z_q^n$
\[d_+(v,w) := \left\{ \begin{array}{ll} |v_i - w_i|  & \textnormal{ if } v_i\neq w_i \textnormal{ and } v_j= w_j \forall j\neq i \\ 0 & \textnormal{ if } v=w \\ \infty & \textnormal{ if } \exists i, j : i\neq j,v_i\neq w_i, v_j\neq w_j \end{array}\right. .\]
\end{defi}
The cross distance is not a proper metric but it is an extended semi-metric, i.e.\ $\infty$ is allowed as a value and the triangle inequality does not hold.

\begin{thm}
The cross distance sphere with center $c$ and radius $t$, $S^+_t(c):=\{v\in \Z_q^n \mid d_+(v,c) \leq t\}$ is exactly the set of $c$ plus all possible cross errors of magnitude at most $t$, i.e.\
\[ S^+_t(c) = \{c+ \alpha e_i \mid |\alpha | \leq t, i\in\{1,\dots,n\}\} .\]
\end{thm}

It follows that a code $C\subseteq \Z_q^n$ is cross-error correcting with error magnitude $t$ if and only if its minimum cross distance $d_+(C):=\min\{d_+(v,w) \mid v,w \in C, v\neq w\}$ is at least $2t+1$.



One can easily count the cardinality of a cross sphere:

\begin{lem}\label{lem5}
A cross sphere in $\Z_{q}^n$ with radius $t $ and any center $c\in \Z_{q}^n$ has volume
\[|S^+_t(c) |=2nt  +1.\]
\end{lem}

We will now derive the sphere packing bound for cross codes in $\Z_{2^m}^n$.

\begin{thm}\label{cor3}
The sphere packing bound for 
cross-error correcting codes $C^+ \subseteq \Z_{2^m}^n$ is given by
\[|C^+|  \leq \frac{|\Z_{2^m}^n |}{|S^+_t(\mathbf 0) |} = \frac{2^{nm}}{2nt  +1}   .\]
For linear codes the cardinality is upper bounded by the greatest power of $2$ that is below the sphere packing bound.
\end{thm}
\begin{proof}
The first statement follows from the previous lemma. 
The second follows, since a linear code is an additive subgroup of $\Z_{2^m}^n$ and thus has a cardinality that divides $2^{mn}$ by Lagrange's Theorem. Hence, $|C^+_{lin}|$ is a power of $2$
and an upper bound on the cardinality is therefore given by the greatest power of $2$ that is less than the respective bound from before. 
\end{proof}

The cardinality of Lee spheres is well-known (see e.g.\ \cite{go70a}) and hence the sphere packing bounds for the Lee metric is
\[|C^L |  \leq \frac{2^{nm}}{\sum_{i=0}^{\min\{n,t\}} 2^i \binom{n}{i} \binom{t}{i}}   .\]
One can easily see that the sphere packing bound for cross-error correcting codes is higher than the one for Lee codes if $t \geq 2$ and they are equal for $t =1$.

The following tables give upper bounds on the size of linear and non-linear Lee and cross-error correcting codes in $ \Z_{2^m}^n$, for magnitude $t $.

\begin{table}[ht]
\begin{center}
\begin{tabular}{|c|c|c|c|c|}
\hline
$2^m$ &$ C^L$ &$ C^+$ & $C^L_{lin}$ & $C^+_{lin}$\\
\hline
8 & 4 & 7 & 4 & 4 \\
16 & 19 & 28 & 16 & 16 \\
32 & 78 & 113 & 64 & 64 \\
\hline
\end{tabular}
\caption{Sphere packing bounds on the cardinality of the different codes  in $ \Z_{2^m}^2$  for $t  =2$.}
\end{center}
\end{table}

\begin{table}[ht]
\begin{center}
\begin{tabular}{|c|c|c|c|c|}
\hline
$2^m$ &$ C^L$ &$ C^+$ & $C^L_{lin}$ & $C^+_{lin}$\\
\hline
8 & 2 & 4 & 2 & 4 \\
16 & 10 & 19 & 8 & 16 \\
32 &  40 & 79 & 32 & 64 \\
\hline
\end{tabular}
\caption{Sphere packing bounds on the cardinality of the different codes  in $ \Z_{2^m}^2$  for $t  =3$.}
\label{table2}
\end{center}
\end{table}

\begin{table}[ht]
\begin{center}
\begin{tabular}{|c|c|c|c|c|}
\hline
$2^m$ &$ C^L$ &$ C^+$ & $C^L_{lin}$ & $C^+_{lin}$\\
\hline
8 & 20 & 39 & 16 & 32 \\
16 & 163 & 316 & 128 & 256 \\
32 & 1310 & 2521 & 1024 & 2048 \\
\hline
\end{tabular}
\caption{Sphere packing bounds on the cardinality of the different codes  in $ \Z_{2^m}^3$ for $t  =2$.}
\end{center}
\end{table}

\begin{table}[ht]
\begin{center}
\begin{tabular}{|c|c|c|c|c|}
\hline
$2^m$ &$ C^L$ &$ C^+$ & $C^L_{lin}$ & $C^+_{lin}$\\
\hline
8 & 8 & 26 & 8 & 16 \\
16 & 65 & 215 & 64 & 128 \\
32 & 520 & 1724 & 512 & 1024 \\
\hline
\end{tabular}
\caption{Sphere packing bounds on the cardinality of the different codes  in $ \Z_{2^m}^3$ for $t  =3$.}
\end{center}
\end{table}

\begin{table}[ht]
\begin{center}
\begin{tabular}{|c|c|c|c|c|}
\hline
$2^m$ &$ C^L$ &$ C^+$ & $C^L_{lin}$ & $C^+_{lin}$\\
\hline
8 & 99 & 240 & 64 & 128 \\
16 & 1598 & 3855 & 1024 & 2048 \\
32 & 25572 & 61680 & 16384 & 32768 \\
\hline
\end{tabular}
\caption{Sphere packing bounds on the cardinality of the different codes  in $ \Z_{2^m}^4$ for $t  =2$.}
\end{center}
\end{table}


A classical question in coding theory is if there exist \emph{perfect codes}, i.e.\ the spheres of a given radius $t$ partition the whole space.

\begin{prop}
There are no perfect cross-error correcting codes over $ \Z_{2^m}$.
\end{prop}
\begin{proof}
We know that $ |\Z_{2^m}^n|= 2^{mn}$ is a power of $2$. By Lemma \ref{lem5} we further know that for any $t\geq 1$, $|S^+_t(c) |$ is not a power of $2$ and does thus not divide $ |\Z_{2^m}^n|$.
\end{proof}



\section{Constructions for linear cross codes}

We will now derive some general constructions for linear cross codes. 
%
For simplicity we will do this separately for code length $n=2$ and $n=3$. The ideas of these constructions can then be used for similar constructions for larger values of $n$.

\subsection{Length $n=2$}

 Let $k:= \max \{i\in \N \mid 2^i \leq t \}$ and $\bar t  = t $ if $t $ is odd and  $\bar t  = t -1$ if $t $ is even.

\begin{thm}\label{thm18}
Let $m\geq k$. 
The following is a  parity check matrix of a cross code in $\Z_{2^m}^2$ with error magnitude $t $:
\[H= \left(\begin{array}{cccc} x_1 & y_1  \\ x_2 & y_2 \end{array}\right)\]
where
\begin{footnotesize}
\[ x_1,y_1 \not \in \pm \{0,2^{m-1},2^{m-2}, \dots, 2^{m-k-1}\} \mod 2^m ,   \]
\[  y_2\not \in  \pm \{1,\dots,t \} \{1,3^{-1},5^{-1}, \dots, \bar t ^{-1}\} x_2 \mod 2^{m}   ,\]
 \[ x_2\not \in  \pm \{1,\dots,t \} \{1,3^{-1},5^{-1}, \dots, \bar t ^{-1}\} y_2 \mod 2^{m}   ,\] 
 \[ y_2\not \in  \pm \{1,\dots,\lfloor\frac{t }{2^k}\rfloor\} \left\{1,3^{-1},5^{-1}, \dots, \bar{ \lfloor\frac{t }{2^k}\rfloor}^{-1}\right\} x_2 \mod 2^{m-k}   ,\]
 \[ x_2\not \in  \pm \{1,\dots,\lfloor\frac{t }{2^k}\rfloor\}\left\{1,3^{-1},5^{-1}, \dots, \bar{ \lfloor\frac{t }{2^k}\rfloor}^{-1}\right\} y_2 \mod 2^{m-k}   .  \]
\end{footnotesize}
\end{thm}
\begin{proof}
Since in a  linear code all differences of two codewords is again a codeword, it is enough to check if all codewords fulfill the non-intersection property with the all zero word. 
Let $(a,b)\in \Z_{2^m}^2$ be a codeword, i.e.\ $(a,b)H^T= 0$. 

Then the first row of $H$ implies that  if $a=0$, then $\pm b> 2t $, and if $b=0$, then $\pm a > 2t $.

Now assume that both $a$ and $b$ are non-zero.The second row of $H$ gives rise to the following parity check equation
\begin{align*}
x_2 a + y_2 b \equiv 0 \mod 2^m  .
\end{align*}
Now if $b\in \{1,3,5, \dots, \bar t \}$, then the previous equation is equivalent to
\[ y_2  \equiv x_2 ab^{-1}  \mod 2^m  ,\]
which implies that $a \not \in  \pm \{1,\dots,t \}$ (follows from (2)). 
In the same way one can see that if $a\in \{1,3,5, \dots, \bar t \}$,  then $b \not \in  \pm \{1,\dots,t \}$ (follows from (3)). 
Now assume that both $a$ and $b$ are divisible by $ 2^{k'}$, where $k'\leq k$ and $b\in \pm \{1,\dots,t \}$. Then we get
\[x_2 a + y_2 b \equiv 0 \mod 2^m  \]
\[\iff x_2 a 2^{-k'}  \equiv - y_2 b2^{-k'} \mod 2^{m-k'} . \]
We can choose $k'$ maximal such that either $a':= a2^{-k'}$ or $ b':= b2^{-k'} $ (or both) is odd and hence invertible. If $ b' $ is odd then we get
\[-x_2 a' {b'}^{-1}  \equiv  y_2 \mod 2^{m-k'}\]
i.e.\ if $b' \in  \pm \{1,3,\dots,\bar{\lfloor\frac{t }{2^{k'}}\rfloor}\}$ (i.e.\ $b \in  \pm \{2^{k'},3\cdot 2^{k'},\dots,\bar{\lfloor\frac{t }{2^{k'}}\rfloor}2^{k'}\}$) , then $a' \not \in  \pm \{1,\dots,\lfloor\frac{t }{2^{k'}}\rfloor\}$, which implies that $a \not \in  \pm \{2^{k'},3\cdot 2^{k'},\dots,\bar{\lfloor\frac{t }{2^{k'}}\rfloor}2^{k'}\}$. Since we assumed that $2^{k'}$ divides $a$ this implies that $|a|>t $. 
Analogously, if $a' \in  \pm \{1,3,\dots,\bar{\lfloor\frac{t }{2^{k'}}\rfloor}\}$ (i.e.\ $a \in  \pm \{2^{k'},3\cdot 2^{k'},\dots,\bar{\lfloor\frac{t }{2^{k'}}\rfloor}2^{k'}\}$) , then $b' \not \in  \pm \{1,\dots,\lfloor\frac{t }{2^{k'}}\rfloor\}$, which implies that $b \not \in  \pm \{2^{k'},3\cdot 2^{k'},\dots,\bar{\lfloor\frac{t }{2^{k'}}\rfloor}2^{k'}\}$. Thus  $|b|>t $. 

Overall none of our non-zero codewords are of the form $(0,a), (a,0)$ where $a\in \pm\{1,\dots, 2t \}$ or  $(a,b)$ where $a,b \in \pm \{1,\dots,t \}$. 
One can easily check that these properties are enough to ensure the non-intersection of the crosses with the all-zero word.
\end{proof}

Note that with the previous construction, a parity check matrix for codes with error magnitude $2^k$ is the same as for magnitude $2^k+1, 2^k+2, \dots, 2^{k+1}-1$. Thus, we can assume that this construction will be most efficient when $t +1$ is a power of $2$.

To make the cardinality as large as possible we want to choose $x_1,x_2,y_1,y_2$ possibly not invertible and to have the possibly highest power of $2$ as a factor.
Note that we can then always choose the first row of $H$ as all $2^{m-k-2}$ -- no other choice of $x_1,x_2$ will result in a code of larger cardinality.

Moreover, we can choose $x_2=0$ and get the following general form of a parity check matrix.

\begin{cor}\label{cor5}
Let $m\geq k+2$. 
The following is a  parity check matrix of a cross code in $\Z_{2^m}^2$ with error magnitude $t $:
\[H= \left(\begin{array}{cccc} 2^{m-k-2} & 2^{m-k-2}  \\ 0 & 2^{m-k-1} \end{array}\right) .\]
The cardinality of this code is 
\[|C|= 2^{2(m-k)-3} .\]
\end{cor}
\begin{proof}
The cardinality can easily be computed from solving the system of equations from $H$. The second row has a solution space of size $2^{m-k-1}$ and for a given solution from that row, the first row has a solution space of size $2^{m-k-2}$. Multiplying these two gives the overall cardinality of the code.
\end{proof}

\begin{rem}
The codes constructed in Corollary \ref{cor5} are also $t$-error correcting codes for the Lee metric.
\end{rem}

\begin{ex}\label{ex5}
We will now derive cross codes with error magnitude $t =3$ with parity check matrices according to Corollary \ref{cor5}:
\begin{enumerate}
\item
Over $\Z_{8}$:
\[H= \left(\begin{array}{cccc} 1&1  \\0&2  \end{array}\right)\]
defines a code of cardinality $2$  with generator matrix
\[G= \left(\begin{array}{cccc} 4& 4    \end{array}\right).\]
\item
Over $\Z_{16}$:
\[H= \left(\begin{array}{cccc} 2& 2  \\ 0& 4  \end{array}\right)\]
defines a code of cardinality $8$ with generator matrix
\[G= \left(\begin{array}{cccc} 4& 4 \\ 0 & 8   \end{array}\right).\]
\item 
Over $\Z_{32}$:
\[H= \left(\begin{array}{cccc} 4  &  4 \\ 0  &  8 \end{array}\right)\]
defines a code of cardinality $32$ with the same generator matrix as  in $2)$.
\end{enumerate}
\end{ex}

Note that the codes from the previous example would be the same when using Corollary \ref{cor5} to construct a code for $t =2$.

\begin{ex}
We will now derive cross codes with error magnitude $t =7$ with parity check matrices according to Corollary \ref{cor5}:
\begin{enumerate}
\item
Over $\Z_{16}$:
\[H= \left(\begin{array}{cccc} 1& 1  \\ 0& 2  \end{array}\right)\]
defines a code of cardinality $2$ with generator matrix
\[G= \left(\begin{array}{cccc} 8& 8    \end{array}\right).\]
\item 
Over $\Z_{32}$:
\[H= \left(\begin{array}{cccc} 2  &  2 \\ 0  &  4 \end{array}\right)\]
defines a code of cardinality $8$ with generator matrix
\[G= \left(\begin{array}{cccc} 8  &  8 \\ 0  &  16 \end{array}\right) .\]
\end{enumerate}
\end{ex}

We will now investigate how far away from the sphere packing bound this code construction is.

\begin{thm}\label{thm13}
The codes constructed according to Corollary \ref{cor5} are a factor $2^{k+1}$ away from the linear sphere packing bound from Theorem \ref{cor3}.
\end{thm}

\begin{proof}
For $n=2$ the sphere packing bound is $\frac{2^{2m}}{4t + 1}$ and the greatest power of $2$ below this bound is $2^{2m-k-2}$. When we divide this by the cardinality formula $2^{2(m-k)-3}$ we get 
\[\frac{2^{2m-k-2}}{2^{2(m-k)-3}} = 2^{k+1} .\]
\end{proof}

This means that these code are asymptotically optimal for growing $m$.

As mentioned before, for $t $ that is a power of $2$ this construction will most likely not be close to optimal. For $t =2$ (and $t=3$) we have the following result.

\begin{thm}\label{thm9}
Let $t\in\{2,3\}$ and $m\geq 2t$. 
The code in $\Z_{2^m}^2$ with parity check matrix
\[H= \left( -(t+1) \cdot 2^{m-2t} \quad 2^{m-2t}\right)\]
or equivalently with generator matrix
\[G=  \left(\begin{array}{cccc} 1& t+1  \\ 16& 0  \end{array}\right)\]
is a cross code with magnitude $t$ and cardinality  $2^{2(m-t)}$. Note that for $m=4$ the second row of $G$ vanishes.
\end{thm}
\begin{proof}
The two entries of $H$ fulfill conditions (1)--(5) from Theorem \ref{thm18} for $t =2,3$, combined in one row. This implies the error correction capability.

The cardinality can be computed by solving the linear equation arising from $H$
\[  -(t+1) \cdot 2^{m-2t} a + 2^{m-2t} b \equiv 0 \mod 2^m\]
\[\iff  (t+1)  a  \equiv b \mod 2^{2t}   .\]
Hence there are $2^m$ choices for $a$, and for each $a$ there are $2^{m-2t}$ choices for $b\in \Z_{2^m}$. This implies the statement. 
\end{proof}

Note that for $t=3$ and $m=5$ the code defined by the generator matrix $G$ from Theorem \ref{thm9} is a cross codes with error magnitude $t$ and cardinality $2^5$. 

We again  investigate how far away from the sphere packing bound this code construction is.

\begin{thm}
For $t=2$, the codes constructed according to Theorem~\ref{thm9} are a factor $2$ away from the sphere packing bound from Theorem~\ref{cor3}.
For $t=3$, the codes constructed according to Theorem~\ref{thm9} are a factor $8$ away from the sphere packing bound from Theorem~\ref{cor3}.
\end{thm}

\begin{proof}
Since $k=1$ for both $t=2$ or $t=3$, the linear sphere packing bound is $2^{2m-3}$ (cf.\ proof of Theorem \ref{thm13}). We divide this by the cardinality $2^{2m-4}$ to get 
\[\frac{2^{2m-3}}{2^{2(m-t)}} = 2^{2t-3} ,\]
which implies the statements.
\end{proof}



\subsection{For length $n=3$}

We will now describe a construction for cross-error correcting codes in $\Z_{2^m}^3$ with magnitude $t$. 
As before  let $k:= \max \{i\in \N \mid 2^i \leq t \}$.

\begin{thm}\label{thm24}
A parity check matrix of the form
\[H= \left(\begin{array}{cccc}  x_1 &  y_1  & z_1 \\  x_2 &  y_2  & z_2  \end{array}\right)\]
where
\begin{enumerate}
\item
$ x_1,y_1,z_1 \not \in \pm \{0,2^{m-1},2^{m-2}, \dots, 2^{m-k-1}\} \mod 2^m ,   $
\item 
 $ \{1,\dots,t \}x_2 \cap  \pm \{1,\dots,t \} y_2 = \emptyset   \mod 2^{m}, $\\
 $ \{1,\dots,t \}x_2 \cap  \pm \{1,\dots,t \} z_2 = \emptyset   \mod 2^{m}, $\\
 $ \{1,\dots,t \}y_2 \cap  \pm \{1,\dots,t \} z_2 = \emptyset   \mod 2^{m}, $
 \item
 $ \{1,\dots,t \}x_2 \cap  \pm \{1,\dots,\lfloor\frac{t }{2^k}\rfloor \} y_2 = \emptyset   \mod 2^{m-k}, $\\
 $ \{1,\dots,t \}x_2 \cap  \pm \{1,\dots,\lfloor\frac{t }{2^k}\rfloor \} z_2 = \emptyset   \mod 2^{m-k}, $\\
 $ \{1,\dots,t \}y_2 \cap  \pm \{1,\dots,\lfloor\frac{t }{2^k}\rfloor \} z_2 = \emptyset   \mod 2^{m-k}, $
 \end{enumerate}
defines a cross-error correcting code in $\Z_{2^m}^3$ of magnitude $t $.
\end{thm}
\begin{proof}
The proof is analogous to the one of Theorem \ref{thm18}, just this time we have to impose the conditions on all possible pairs of $x_2,y_2,z_2$.
\end{proof}

\begin{cor}\label{cor12}
Assume that $t \leq 2^{m-1}$ (otherwise a cross of this magnitude cannot be defined). 
A parity check matrix of the form
\[H= \left(\begin{array}{cccc}  2^{m-k-2} & 2^{m-k-2} & 2^{m-k-2} \\0 & 2^{m-k-1} & (2t +1)\cdot 2^{m-k-2}  \end{array}\right), \] 
defines a cross-error correcting code in $\Z_{2^m}^3$ of magnitude $t $.
\end{cor}
\begin{proof}
The proof is analogous to before.
\end{proof}

\begin{ex}
We will now derive cross codes with error magnitude $t =3$ with parity check matrices according to Corollary \ref{cor12}:
\begin{enumerate}
\item
Over $\Z_{16}$:
\[H= \left(\begin{array}{cccc} 2& 2&2  \\ 0& 4 & -2  \end{array}\right)\]
defines a code of cardinality $ 64$ with generator matrix
\[G= \left(\begin{array}{cccc} 2&2&4 \\ 1&5&2   \end{array}\right).\]
\item 
Over $\Z_{32}$:
\[H= \left(\begin{array}{cccc} 4& 4&4  \\ 0& 8 & -4  \end{array}\right)\]
defines a code of cardinality $ 512$ with the same generator matrix as in $1)$.
\end{enumerate}
\end{ex}

\begin{ex}
We will now derive cross codes with error magnitude $t =7$ with parity check matrices according to Corollary \ref{cor12}:
\begin{enumerate}
\item
Over $\Z_{16}$:
\[H= \left(\begin{array}{cccc} 1& 1&1  \\ 0& 2 & -1  \end{array}\right)\]
defines a code of cardinality $ 16$ with generator matrix
\[G= \left(\begin{array}{cccc} 7&3&6 \\ 1&5&10   \end{array}\right).\]
\item 
Over $\Z_{32}$:
\[H= \left(\begin{array}{cccc} 2& 2&2  \\ 0& 4 & -2  \end{array}\right)\]
defines a code of cardinality $ 128$ with the same generator matrix as in $1)$.
\end{enumerate}
\end{ex}


\section{Decoding}

We will now explain how these linear codes can be decoded with a syndrome decoder.

\begin{lem}
Assume that the error vector $e\in \Z_{2^m}^n$ has only one non-zero coordinate $i$ (i.e.\ Hamming weight $1$) whose value $\alpha$ is in $\pm \{1,\dots,t \}$. I.e.\ $e= \alpha e_i$, where $e_i$ is the $i$-th unit vector. Then the syndrome vector
\[s=rH^T = (c+e)H^T = eH^T\]
is the $\alpha$-multiple of the transpose of the $i$-th column of $H$.
\end{lem}

Hence, if we can easily identify the multiples of the columns of $H$, we can easily syndrome decode our codes. In fact, this can be done for the parity check matrices described in the previous section. We will describe some decoding algorithms for the various previously explained constructions in Algorithms \ref{alg2} -- \ref{alg3}.

We will start with the algorithm for the codes from Corollary \ref{cor5}. In this case one can easily distinguish the two columns of $H$ because of the zero entry. The algorithm is described in Algorithm \ref{alg2}.

\begin{algorithm}
\begin{algorithmic}
\REQUIRE{Received vector $r \in \Z_{2^m}^2$.}
\STATE Compute the syndromes $(s_1 \; s_2)= r H^T$.
\IF{$s_2=0$}
\IF{ $2^{m-k-2} | s_1$} 
\STATE $e:=( s_1/2^{m-k-2} \; 0)$
\ELSE \RETURN failure
\ENDIF
\ELSIF{ $2s_1=s_2$}
\IF{ $2^{m-k-2} | s_1$} 
\STATE $e:=(0 \; s_1/2^{m-k-2})$
\ELSE \RETURN failure
\ENDIF
\ELSE \RETURN failure
\ENDIF
\RETURN $c=r-e$
\end{algorithmic}
\caption{Decoding Algorithm for Codes in $\Z_{2^m}^2$ constructed according to Corollary \ref{cor5}.}
\label{alg2}
\end{algorithm}

\begin{ex}
Consider the code from Example \ref{ex5} over $\Z_{16}$ and a received word $r=(12 \quad 6)$. Then
\[(s_1 \quad s_2) = rH^T = (4 \quad 8) ,\]
i.e.\ $2s_1 = s_2$ which means that the error is of the form
\[e=(0\quad s_1 / 2) = (0 \quad 2) .\]
Hence, we decode to the codeword 
\[c= r-e = (12\quad 4)  .\]
\end{ex}

Next we describe an algorithm for the codes from Theorem \ref{thm9} for $t=2$. In this case we only have one row for the parity check matrix, so we would have to distinguish if the syndrome is a multiple of $3\cdot 2^{m-4}$ or of $2^{m-4}$, which is in general not possible since $3$ is invertible over $\Z_{2^m}$. In our case though, we assume that the error value is in $\pm\{1,2\}$, hence the syndrome is equal to $\pm 3\cdot 2^{m-4} $ if $e=(1 \; 0)$, to $\pm 3\cdot 2^{m-3} $ if $e=(2 \; 0)$,  to $\pm  2^{m-4} $ if $e=(0 \; 1)$, and  to $\pm 2^{m-3} $ if $e=(0 \; 2)$.
The algorithm is described in Algorithm \ref{alg1}. Note that the variables $i$ and $j$ can take values $0$ and $1$ only.

\begin{algorithm}
\begin{algorithmic}
\REQUIRE{Received vector $r \in \Z_{2^m}^2$.}
\STATE Compute the syndrome $s= r H^T$.
\IF{ $ \exists i,j \in \{0,1\}: s=(-1)^i 3\cdot 2^j \cdot 2^{m-4}$}
\STATE $e:=((-1)^i 2^j \quad 0)$
\ELSIF{ $ \exists i,j \in \{0,1\}: s=(-1)^i 2^j 2^{m-4}$}
\STATE $e:=(0 \quad (-1)^i 2^j )$
\ELSE \RETURN failure
\ENDIF
\RETURN $c=r-e$
\end{algorithmic}
\caption{Decoding Algorithm for Codes in $\Z_{2^m}^2$ constructed according to Theorem \ref{thm9} for $t=2$.}
\label{alg1}
\end{algorithm}

Last we describe an algorithm for the codes of length $3$ from  Corollary \ref{cor12}, which is similar to Algorithm  \ref{alg1}.

\begin{algorithm}
\begin{algorithmic}
\REQUIRE{Received vector $r \in \Z_{2^m}^3$.}
\STATE Compute the syndromes $(s_1 \; s_2)= r H^T$.
\IF{$s_2=0$}
\IF{ $2^{m-k-2} | s_1$} 
\STATE $e:=( s_1/2^{m-k-2} \; 0 \; 0)$
\ELSE \RETURN failure
\ENDIF
\ELSIF{ $2s_1=s_2$}
\IF{ $2^{m-k-2} | s_1$} 
\STATE $e:=(0 \; s_1/2^{m-k-2} \; 0)$
\ELSE \RETURN failure
\ENDIF
\ELSIF{ $(2t  +1)s_1=s_2$}
\IF{ $2^{m-k-2} | s_1$} 
\STATE $e:=(0 \; 0\; s_1/2^{m-k-2})$
\ELSE \RETURN failure
\ENDIF
\ELSE \RETURN failure
\ENDIF
\RETURN $c=r-e$
\end{algorithmic}
\caption{Decoding Algorithm for Codes in $\Z_{2^m}^3$ constructed according to Corollary \ref{cor12}.}
\label{alg3}
\end{algorithm}


\section{Conclusion}

In this work we investigated cross- error correcting integer codes. We presented a metric model that represents this type of errors and derive some theoretical results like the sphere packing bound for this metric. Then we derived code constructions for cross-error correcting codes of magnitude $t$ in $\Z_{2^m}^2$ and $\Z_{2^m}^2$ for general $m$ and $t$.  The respective codes asymptotically attain the sphere packing bound for growing $m$. Furthermore, we presented efficient decoding algorithms for these constructions.

In future research we would like to see if these code constructions are optimal, i.e.\ either find a tighter bound for linear cross codes or find a larger code for a given set of parameters. Moreover, we would like to derive a construction for general code length $n$.


\bibliographystyle{plain}
\bibliography{integer_codes_stuff}

\end{document}